\documentclass[aps,pra,twocolumn,amsmath,amssymb,nofootinbib,superscriptaddress]{revtex4}

\usepackage[colorlinks]{hyperref}
\usepackage[pdftex]{graphicx}
\usepackage{mathrsfs}
\usepackage{amsthm}
\usepackage{float}
\usepackage{comment}
\usepackage{cancel}

\newtheorem{proposition}{Proposition}

\newtheorem{definition}{Definition}

\usepackage[usenames,dvipsnames]{xcolor}
\usepackage{framed}
\usepackage{soul}

\definecolor{cream}{rgb}{1.0, 0.99, 0.82}
\definecolor{celadon}{rgb}{0.67, 0.88, 0.69}
\definecolor{beaublue}{rgb}{0.74, 0.83, 0.9}
\definecolor{shadecolor}{rgb}{1.0, 0.99, 0.82}
\definecolor{agcol}{rgb}{0.9, 1.0, 0.9}
\definecolor{jkcol}{rgb}{0.3, 1.0, 0.5}

\newcommand{\bra}[1]{\left\langle#1\right\rvert}
\newcommand{\ket}[1]{\left\lvert#1\right\rangle}
\newcommand{\ketbra}[2]{\left\lvert{#1}\middle\rangle\!\middle\langle{#2}\right\rvert}
\newcommand{\braket}[2]{\left\langle{#1}\middle\vert{#2}\right\rangle}

\newcommand{\tr}{\mathrm{Tr}}

\begin{document}

\bibliographystyle{apsrev}

\title{Completeness Stability of Quantum Measurements}
\author{Rakesh Saini}
\email[]{rakesh7698saini@gmail.com}
\affiliation{Centre for Engineered Quantum Systems, Department of Physics and Astronomy, Macquarie University, Sydney NSW 2113, Australia}
\author{Jukka Kiukas}
\affiliation{Department of Mathematics, Aberystwyth University, Aberystwyth SY23 3BZ, United Kingdom}
\author{Daniel Burgarth}
\affiliation{Physics Department, Friedrich-Alexander Universit\"at of Erlangen-Nuremberg, Staudtstr. 7, 91058 Erlangen, Germany}
\affiliation{Department of Mathematics, Aberystwyth University, Aberystwyth SY23 3BZ, United Kingdom}
\affiliation{Centre for Engineered Quantum Systems, Department of Physics and Astronomy, Macquarie University, Sydney NSW 2113, Australia}
\author{Alexei Gilchrist}
\affiliation{Centre for Engineered Quantum Systems, Department of Physics and Astronomy, Macquarie University, Sydney NSW 2113, Australia}

\date{\today}

\begin{abstract}
We introduce a resource monotone, the \emph{completeness stability}, to quantify the quality of quantum measurements within a resource-theoretic framework. By viewing a quantum measurement as a frame, the minimum eigenvalue of a frame operator emerges as a significant monotone.
It captures bounds on estimation errors and the numerical stability of inverting the frame operator to calculate the optimal dual for state reconstruction.
Maximizing this monotone identifies a well-characterized class of quantum measurements forming weighted complex projective 2-designs, which includes well-known examples such as SIC-POVMs. Our results provide a principled framework for comparing and optimizing quantum measurements for practical applications.
\end{abstract}

\maketitle
    
\section{Introduction}

Measurement devices constitute a fundamental component of quantum mechanics, serving as the mechanism for extracting information from a quantum system. The quality of the measurement device depends on various factors, including the mechanisms that create correlations between the system being measured and measurement outcomes, environmental noise, and the specific characteristics of the apparatus. Quantifying the quality through a general performance-based measure is therefore a task of significant importance.

The information-gain aspect of quantum measurements is captured by a Positive Operator-Valued Measure (POVM) \cite{1996BusQTM,2011HolPsa}. A POVM on a \( d \)-dimensional Hilbert space is \emph{informationally complete} if it consists of at least \(  d^2 \) effects that span the space of linear operators. Such measurements can in principle uniquely reconstruct any quantum state in the limit of infinite copies. In practice however, only a finite number of copies are available. Therefore, the quality of a measurement should be evaluated by how accurately and reliably it reconstructs an arbitrary quantum state under this constraint. A measurement device that facilitates efficient and accurate reconstruction with limited data is deemed superior.

To rigorously quantify the quality of a measurement device, we adopt a \emph{resource-theoretic} framework~\cite{2019chit001}. Resource theories provide a structured approach to quantifying the value of a given resource based on the utility in a given task. Resource theories can be defined by specifying a set of \emph{free} objects and \emph{free} transformations—operations that cannot increase the resource. Every object and operation outside this set is resourceful. The value of a resource can then be characterized by resource \emph{monotones}—real-valued functions that do not increase under the free operations. These monotones provide a consistent way to compare and order different objects based on their resourcefulness.
Resource-theoretic techniques have found broad application in quantum mechanics, including the study of entanglement, thermodynamics, coherence, reference frames, and more~\cite{2009hor942,2016goo001,2007bar609,2017str003,2019chit001}. Within quantum measurements, several resource theories have been developed based on distinct features: measurement sharpness~\cite{2024Bus235}, coherence~\cite{2020bae019}, incompatibility~\cite{2020Bus401,2015hkr115, 2022kmp205}, and post-processing~\cite{2021guf301,2019skr403}.

Our work focuses on resource theories based on classical post-processing. In this setting, the free operations correspond to stochastic mixing of POVM effects. The measurement statistics of mixtures can always be simulated by the original POVM followed by classical post-processing, but not necessarily vice versa. Consequently, any post-processed POVM cannot be more resourceful than the original POVM in any information-theoretic task such as state discrimination, tomography, or parameter estimation. Any valid resource monotone must therefore be non-increasing under post-processing.
Ideally, one would like a \emph{complete} set of monotones—quantities that fully characterize this ordering—such that whenever all monotones in the set assign greater or equal value to one POVM over another, it implies the former can simulate the latter via post-processing. While the set of all quantum state discrimination tasks constitutes such a complete family~\cite{2019skr403}, it is unwieldy and impractical to work with. Thus, identifying a small set of physically well motivated measurement monotones remains an important goal in the resource theory of quantum measurements.
One well-established monotone in this framework is the \emph{robustness of measurement}~\cite{2019skr403}, which quantifies the amount of noise required to render a measurement device completely uninformative, i.e., when each effect becomes proportional to the identity. However, in the regime of informationally complete POVMs, this monotone does not adequately capture reconstruction performance. This motivates the need for new monotones tailored to this regime.

To develop more suitable measures for quantifying the resourcefulness of informationally complete POVMs, finite frame theory provides a natural mathematical foundation~\cite{2006sco507, 2007dar403, 2011zhu327, 2014zhu115, 2022Per851, 2004Ari487}. A \emph{frame} is a set of vectors that spans a Hilbert space but which need not be linearly independent~\cite{2018walint, 2008kov094,2013Cas053}. In the context of quantum measurements, any informationally complete POVM provides a natural frame, namely the collection of its effects considered as vectors.
In frame theory, a \emph{frame operator} is used for reconstruction of vectors in terms of inner-products with the original frame. We find instead that a related construction where  each POVM element is scaled by the square root of its trace, leads to a powerful resource monotone. This \emph{scaled frame operator}, naturally appears across a variety of quantum information tasks. Notably, in the context of \emph{shadow tomography}~\cite{2017aar053}, it has been shown that the average state estimation error as well as average variance in estimating the expectation value of observables is minimized when using the scaled frame operator, assuming no prior information about the state~\cite{2023inn328}.
Moreover, these structures are deeply connected to optimal measurement strategies in quantum tomography and cloning, through an equivalence with complex projective 2-designs~\cite{2006sco507}.
In addition to these applications, the scaled frame operator is also used to study the tomographic efficiency using various estimation methods~\cite{ 2006sco507, 2011zhu327, 2014zhu115,2014zhu309}.
 
We show that this scaled frame operator fits naturally within the resource theory of POVMs. This leads to a set of measurement resource monotones that effectively quantify the quality of a measurement device. Among them, one that is particularly useful is the \emph{minimum eigenvalue} of the scaled frame operator, which we term the \emph{completeness stability}. 
We demonstrate that the completeness stability monotone is not only meaningful from a resource-theoretic standpoint but also has important operational significance. As a spectral quantity, it serves to quantify how close a POVM is to being informationally \emph{incomplete}. This has practical consequence that the inverse of the completeness stability is the condition number of the scaled frame operator, which determines the numerical stability of state reconstruction. A higher completeness stability implies lower sensitivity to noise.

Furthermore, in the context of shadow tomography, we show that completeness stability provides a tight upper bound on both the average state estimation error and the variance in estimating expectation values of observables when no prior information about the state is available. Thus, POVMs that maximize the completeness stability lead to more accurate and stable quantum tomography and parameter estimation.
Remarkably, when this monotone is maximized, the resulting POVMs lie on a well-defined boundary characterized by complex projective 2-designs. This class includes some of the most symmetric and optimal measurements known, such as symmetric informationally complete POVMs (SIC-POVMs) and mutually unbiased bases (MUBs), which are widely regarded as optimal for quantum state tomography and quantum cloning tasks.

The structure of the paper is as follows. Section~\ref{sec:background}, introduces the essential tools for our analysis, namely finite frame theory and measurement resource theory. In Section~\ref{sec:scaled-frame} we define the scaled frame operator and discuss requirements for resource monotones based on this operator. In Section~\ref{sec:stability}, we focus on the completeness stability, highlighting its properties and applications. Finally, in Section~\ref{sec:examples} we present some examples.

\section{Background}
\label{sec:background}

\subsection{Frame theory and informationally complete quantum measurements}
\label{sec:frame-theory}

We study quantum measurements using finite \emph{frame theory}~\cite{2018walint, 2008kov094,2013Cas053}. Frame theory provides a robust method for representing and decomposing arbitrary vectors in terms of a spanning set of vectors called a \emph{frame}, which may contain linearly dependent vectors. This redundancy in the frame offers greater flexibility and resilience in reconstruction, particularly under additive and independent noise. This makes frame theory a natural tool for analyzing the accuracy and reliability of informationally complete POVMs. 

Formally, a frame is a set of vectors $\{\ket{f_{j}}\}$ in a \(d\)-dimensional complex vector space $\mathbb C^d$, and constants $0 < a \le b < \infty$ such that
\begin{align}\label{frame bounds}
  a \braket{x}{x} \le \sum_{j=1}^{n} |\braket{x}{f_{j}}|^{2}  \le b \braket{x}{x} \quad \forall \ket{x} \in \mathbb C^d.
\end{align}
The constants $a$ and $b$ are referred to as the lower and upper frame bounds, respectively. If \(a=b\) the frame is called \emph{tight}. An orthonormal basis forms a tight frame with a bound of 1.
The associated \emph{frame operator} is defined as $F = \sum_{j=1}^{n} \ketbra{f_j}{f_j}$, which by construction is a positive-definite operator, and it has minimum and maximum eigenvalues corresponding to the optimal frame lower and upper bounds respectively. Since \(a>0\), the frame operator is invertible. A set of \emph{dual vectors} $\ket{f'_{j}}$, which also forms a frame, can be defined such that \(\sum_{i} \ket{f'_{j}}\bra{f_{j}}= I\),
which allows the reconstruction of any other vector $\ket{x}$ in terms of the inner-products with the frame vectors: $\ket{x} = \sum_{j} |f'_{j}\rangle\braket{f_{j}}{x} $.
Note that dual frames are not unique, the possible dual frames can be specified \cite{1995li181} as
\begin{equation}
  \ket{f'_{j}} = F^{-1}\ket{f_{j}} + \ket{g_{j}} - \sum_{k}\ketbra{g_{k}}{f_{k}}F^{-1}\ket{f_{j}},
\end{equation}
where the vectors \(\{\ket{g_{j}}\}\) are arbitrary.
Another useful construction that is complementary to the frame operator is the \emph{Gramian}. The Gramian is useful for obtaining the spectral properties of the frame operator. The Gramian is defined component wise as \( [G]_{r,c} = \braket{f_{r}}{f_{c}}\). 
 
We now consider a quantum setting with a $d$-dimensional Hilbert space $\mathcal H$, and let $\mathcal L(\mathcal H)$ denote the set of linear operators on $\mathcal H$ ($d\times d$ matrices). A POVM $\mathcal{A}$ associates a measurement outcome $a$ with a positive
semi-definite operator $A_{a}$, referred to as an \emph{effect}. The
effects satisfy the completeness relation $\sum_{a}A_{a}=I$, where $I$ is the identity operator. Then a quantum state is represented by a positive semidefinite matrix $\rho$ on $\mathcal H$ of unit trace, and the probability of obtaining outcome \(a\) of the POVM is \(p(a|\rho )=\mathrm{Tr}\left[\rho A_{a}\right]\). 

We equip the linear space $\mathcal L(\mathcal H)$ with the Hilbert-Schmidt inner product $\langle A|B\rangle = {\rm Tr}[A^{\dagger}B]$; then we can write $p_{a}=\braket{A_{a}}{\rho}$ where $\ket{\rho},\ket{A_{a}} \in \mathcal{L}(\mathcal H)$, and regard $\mathcal L(\mathcal H)$ as a Hilbert space. Also, $\ket{A_{a}}$ is the \emph{vectorization} of the POVM effect $A_{a}$, constructed as $\ket{A_{a}} \equiv\sum_{j=1}^{d} I\otimes A_{a} \ket{j}\ket{j}$ where $\{\ket{j}\}$ is an orthonormal basis. This is just an explicit representation of the linear operator \(A_{a}\) as a vector in a basis of the system on which it acts.

If the POVM effects span $\mathcal{L}(\mathcal H)$, then the POVM is \emph{informationally complete} and the effects form a frame. In this case we can construct the frame operator and Gramian
\begin{align}
  \label{eq:frame-operator}
  F(\mathcal{A}) &= \sum_{a} \ketbra{A_{a}}{A_{a}} \\
  G(\mathcal{A}) &= \sum_{a,b} \braket{A_{a}}{A_{b}}\ketbra{a}{b},
\end{align}
for the Hilbert space $\mathcal L(\mathcal H)$. The corresponding canonical dual frame is $\ket{A'_{a}} = F(\mathcal{A})^{-1}\ket{{A}_{a}}$ and it is intimately connected to quantum tomography, allowing the reconstruction of a quantum state from observed measurement probabilities:
\begin{align}\label{eq:tomo}
    \sum_{a} p_{a} \ket{A'_{a}} &= \sum_{a} \braket{A_{a}}{\rho} F^{-1} \ket{A_{a}} \nonumber\\
               &= F^{-1} \sum_{a} \ketbra{A_{a}}{A_{a}} \ket{\rho} = \ket{\rho}.
\end{align}
This process is sometimes called ``linear inversion'' tomography to distinguish it from more statistically sophisticated schemes such as maximum likelihood or Bayesian tomographic reconstruction (See discussion in \cite{2010blu034} and references therein). 
Note that a frame is a more general object than a POVM in that its possible to have a set of linear operators that span their vector space that do not form a POVM. Such a set may also be used to reconstruct a state if the inner products in \eqref{eq:tomo} are known.

The relationship between the frame operator and the Gramian is insightful from the perspective of quantum systems. They can be both thought of as reduced states of a larger bipartite pure state. Given a POVM $\mathcal{A}$ with \(n\) effects on a \(d\)-dimensional Hilbert space, first construct the matrix
\begin{equation}\label{eq:Mv}
M(\mathcal{A}) = \frac{1}{N_{\mathcal A}}\sum_{a,a'}\ketbra{a}{a'}\otimes \ketbra{A_{a}}{A_{a'}},
\end{equation}
where $\{\ket{a}\}$ is a basis in \(n\ge d^{2}\) dimensions.
The normalization constant is given by
\begin{align}
N_{\mathcal A}=\tr[F(\mathcal{A})]
=\sum_{a} \braket{A_{a}}{A_{a}}
=\sum_{a}\tr[A_{a}^{\dagger}A_{a}].
\end{align}
The matrix $M(\mathcal{A})$ is a rank-1 projector.
In fact, we clearly have
\begin{equation}
M(\mathcal A) = |K\rangle \langle K|,
\end{equation}
where $|K\rangle$ is a bi-partite pure state across two systems of dimensions $n$ and $d^2$, respectively given by
\begin{equation}
|K\rangle =\frac{1}{\sqrt{N_\mathcal A}}\sum_a |a\rangle \otimes |A_a\rangle.
\end{equation}
Consequently the reduced states of $|K\rangle$ will have the same non-zero spectrum by the Schmidt decomposition. They are $F_{s}(\mathcal{A})$ and $G(\mathcal{\mathcal A})$ up to the factor $N_{\mathcal A}$:
\begin{align}
\tr_{1}[M(\mathcal{A})]
&= \frac{1}{N_{\mathcal A}}\sum_{a} \ketbra{ A_{a}}{A_{a}} = \frac{1}{N_{\mathcal A}}F(\mathcal{A}), \\
\tr_{2}[M(\mathcal{A})]
&= \frac{1}{N_{\mathcal A}}\sum_{a,a'} \braket{ A_{a}}{A_{a'}}\ketbra{a}{a'} =\frac{1}{N_{\mathcal A}}G(\mathcal{A}).
\end{align}

\subsection{Resource theory of POVMs}
In the resource theory of POVMs~\cite{2021guf301}, a POVM $\mathcal{A}$ is defined to be at least as informative as another POVM $\mathcal{B}$, denoted $\mathcal{A} \succeq \mathcal{B}$, if there exists a column-stochastic matrix \( S \) that transforms the effects of $\mathcal{A}$ to those of $\mathcal{B}$:
\begin{equation} \label{eq:B=SA}
  \vec{B} = S\vec{A},
\end{equation}
where for convenience we will represent a POVM with effects $\{A_{1},\ldots,A_{m}\}$ as the symbolic vector $\vec{A}$.
This implies that each effect \( B_b \) of $\mathcal{B}$ can be expressed as a probabilistic combination of the effects \( A_a \) of $\mathcal{A}$:
\begin{equation}
  B_{b} = \sum_{a} p(b|a)A_{a},
\end{equation}
where the probabilities \( p(b|a) \) satisfy \( \sum_{b} p(b|a) = 1 \) for all \( a \).
In essence, \( B_b \) can be interpreted as selecting \( A_a \) with probability \( p(b|a) \), which indicates that any statistical predictions achievable with POVM \( \mathcal{B}  \) are also achievable with POVM \(\mathcal{A} \) with post-processing, but not necessarily the reverse.

As demonstrated in~\cite{2021guf301}, any stochastic matrix \( S \) can be realized using two fundamental operations that \emph{apriori} can not improve a measurement:
\begin{enumerate}
    \item \textbf{Making up outcomes:} Reporting multiple different outcomes from a single outcome, such as \( \{ A_{1} \} \rightarrow \{ p_1 A_{1}, \dots, p_j A_{1} \} \), where \( \sum_{n=1}^{j} p_{n} = 1 \). This operation is reversible since it preserves all information—--it merely partitions the original outcome. As such, it does not reduce the quality of the measurement.
    \item \textbf{Confusing outcomes:} Combining two or more outcomes and reporting as a single outcome, such as \( \{ A_1, A_2 \} \rightarrow \{ A_1 + A_2 \} \). This operation is irreversible and results in a loss of information, thereby degrading the quality of the measurement device.
\end{enumerate}

To quantify the quality of a POVM we define \emph{resource monotones}---functions that assign a numerical value to each POVM in a way that allows comparison across different POVMs. A monotone $\mu$ for this resource theory is any function that maps a POVM to a real number and satisfies:
\begin{equation}\label{monotone equation}
  \mathcal{A} \succeq \mathcal{B} \implies \mu(\mathcal{A}) \ge \mu(\mathcal{B}).
\end{equation}

A well-known example of a resource monotone for POVMs is the \emph{robustness} of a measurement~\cite{2019skr403}, which quantifies the amount of noise that must be added to render the measurement completely uninformative. The robustness of a measurement serves as an important measure of quality, which has recently been extended to general resource theories~\cite{2019tak053}, establishing connections to single-shot information theory and advantages in state discrimination tasks~\cite{2019skr403}.  

In discrimination games, robustness reflects how effectively a quantum measurement device distinguishes states compared to random guessing. Its minimum occurs for measurements proportional to the identity (i.e. uninformative measurements) and its maximum occurs for any set of at least $d$ rank-one projector effects. In fact, any informationally complete POVM composed entirely of rank-one effects achieves this maximal robustness, regardless of the spread or orientation of the effects. As we will illustrate through example later, this insensitivity to the geometric arrangement of the POVM makes robustness an inadequate metric for comparing measurements \emph{within} the informationally complete regime. In this work, we introduce new monotones and explore their effectiveness across this regime. We will see in the next section how using a scaled frame leads to interesting resource monotones.

\section{Scaled Frame}
\label{sec:scaled-frame}

The frame naturally associated with an informationally complete POVM does not always lead to a meaningful resource monotone (see Appendix~\ref{frame monotone}). Scaling each POVM element by the square root of its trace forms a new frame, and we define the \emph{scaled frame operator}, which is a more operationally relevant object.

\begin{definition}
The scaled frame operator for POVM $\mathcal{A}$ is defined as:
\begin{equation}
  F_{s}(\mathcal{A}) = \sum_{a} \frac{\ketbra{A_{a}}{A_{a}}}{\braket{A_a}{I}}.
\end{equation}
\end{definition}
As with the normal frame operator \( F \), we have the upper and lower bounds on the scaled frame operator:
\begin{equation} \label{scaled frame operator bound}
  \Lambda_{\min}[F_{s}(\mathcal{A})] \leq  \bra{X} F_s(\mathcal{A}) \ket{X}  \leq \Lambda_{\max}[F_{s}(\mathcal{A})],
\end{equation}
holding for any operator $X$ with $\langle X| X\rangle =1$, and where $\Lambda_{min}$ and $\Lambda_{max}$ are the minimum and maximum eigenvalues of scaled frame operator respectively. 
The scaled frame operator exhibits interesting properties for POVMs. Notably, all its eigenvalues lie in the range $[0,1]$. The maximum eigenvalue is always $1$ with the identity as an eigenvector.
\begin{proposition}
All eigenvalues of scaled frame operator for POVMs lie in $\left[0,1 \right]$. The maximum eigenvalue is always $1$ with the identity matrix $I$ as an eigenvector.
\end{proposition}

\begin{proof}
  Given a POVM $\mathcal{A}$ and $\ket{X} \in \mathcal{L}(\mathcal H) $:
  \begin{align*}
     \bra{X} F_{s}(\mathcal{A}) \ket{X} & = \sum_{a}  \frac{|\braket{X}{A_a}|^{2}}{\braket{A_a}{I}} = \sum_{a}  \frac{|\mathrm{Tr}[X^\dagger A_a]|^{2}}{\braket{A_a}{I}} \\
   &=  \sum_{a}  \frac{|\mathrm{Tr}[(\sqrt{A_a}X)^\dagger \sqrt{A_a}  ]|^{2}}{\braket{A_a}{I}}  \\
   & \leq \sum_{a}  \frac{\mathrm{Tr}[ \sqrt{A_a} \sqrt{A_a}  ] \,\mathrm{Tr}[X^\dagger \sqrt{A_a}\sqrt{A_a} X  ]}{\braket{A_a}{I}} \\
   &= \sum_{a}  \frac{\mathrm{Tr}[ A_a  ]\,\mathrm{Tr}[X^\dagger  A_a X ]}{\braket{A_a}{I}}\\
   &= \sum_{a}  \mathrm{Tr}[X^\dagger A_a X  ] = \braket{X}{X},
\end{align*}
\noindent
where we applied the Cauchy--Schwarz inequality to the numerator, \( |\braket{X}{Y}|^2 \leq \braket{X}{X} \braket{Y}{Y} \), with equality if and only if \( \ket{X} \propto \ket{Y} \). In our case, this translates to the condition \( \sqrt{A_a} X = \alpha \sqrt{A_a} \) for some scalar \( \alpha \), for every effect \( A_a \). Equivalently, this means \( A_a X = \alpha A_a \) for all \( a \). Since the POVM \( \{A_a\} \) is informationally complete, the set \( \{\ket{A_a}\} \) spans the entire operator space \( \mathcal{L}(\mathcal{H}) \), implying that the only operator \( X \) satisfying this condition for all combination of  \( A_a \) is proportional to the identity. Hence, the inequality is saturated if and only if \( X \propto I \). Using the completeness relation \( \sum_a A_a = I \), we finally verify that:
\begin{equation*}
  F_{s}(\mathcal{A}) \ket{I}= \sum_{a} \ket{A_a} \frac{\braket{A_a}{I}}{\braket{A_a}{I}}= \ket{I}.
\end{equation*}
Hence, the identity matrix is an eigenvector of the scaled frame operator corresponding to the maximum eigenvalue, which is always $1$. 
\end{proof}

These spectral properties offers valuable insight into the reconstruction and estimation of density matrices which we explore further in later sections.

In the resource-theoretic setting, the scaled frame operator remains invariant under reversible operations like making up outcomes, reflecting that such transformations do not degrade the quality of the measurement. In contrast, irreversible operations like confusing outcomes degrade the measurement, yielding a family of monotones that quantify the degradation caused by post-processing, as we prove below.
\begin{proposition}\label{monotones}
  The scaled frame operator provides a set of resource monotones. In fact, for any $\ket{X} \in \mathcal{L}(\mathcal{H})$ the function $\mathcal A\mapsto \langle X|F_s(\mathcal A)|X\rangle$ is a resource monotone.
\end{proposition}

\begin{proof}
Let $\mathcal A$, $\mathcal B$ be POVMs such that $\mathcal A\succeq \mathcal B$, that is,
\[
B_b = \sum_a \mu(b|a)A_a
\]
for all $b$, where $\mu(b|a)$ is some postprocessing. Fix any matrix $X$; we need to show that
$\langle X|F_s(\mathcal B)|X\rangle \leq \langle X|F_s(\mathcal A)|X\rangle$
where e.g.
\[
  \langle X|F_s(\mathcal B)|X\rangle=\sum_b \frac{|\langle B_b|X\rangle|^2}{\langle B_b|I\rangle}.
\]
We first note that 
\begin{align*}
|\langle B_b|X\rangle|^2
&=\left| \sum_a \sqrt{\mu(b|a)\langle A_a|I\rangle}\sqrt{\mu(b|a)}\frac{\langle A_a|X\rangle}{\sqrt{\langle A_a|I\rangle}}\right|^2\\
&\leq \left(\sum_a \mu(b|a) \langle A_a|I\rangle\right)\left(\sum_a\mu(b|a)\frac{|\langle A_a|X\rangle|^2}{\langle A_a|I\rangle}\right).
\end{align*}
Where the initial sum can be seen as an inner-product of two vectors indexed by \(a\), then the Cauchy-Schwartz inequality was applied.
But here $\sum_{a} \mu(b|a) \langle A_a|I\rangle=\langle B_b|I\rangle> 0$, so using the fact that $\sum_b \mu(b|a) = 1$ for any $a$, we get
\begin{align*}
\sum_b \frac{|\langle B_b|X\rangle|^2}{\langle B_b|I\rangle}&\leq \sum_{a,b} \mu(b|a)\frac{|\langle A_a|X\rangle|^2}{\langle A_a|I\rangle}
= \sum_a\frac{|\langle A_a|X\rangle|^2}{\langle A_a|I\rangle},
\end{align*}
which proves the claim.
\end{proof}

A simple consequence of proposition~\ref{monotones} is that if \(\mathcal{A}\succeq \mathcal{B}\) then \( F_s(\mathcal{A}) - F_s(\mathcal{B}) \ge 0 \) is positive semi-definite.
The set of monotones defined by \( \langle X|F_s(\mathcal A)|X\rangle \) are not a complete set, see Appendix~\ref{sec:counter}, but among these monotones, certain choices of \( X \) hold particular significance and offer valuable physical applications.

\section{Completeness Stability for POVMs}
\label{sec:stability}

We now introduce a specific monotone, also based on the scaled frame operator, which captures a measure of the stability of an informationally complete measurement by quantifying how far the scaled frame operator is from invertibility. To be invertible, the scaled frame operator has to be positive definite. The \textit{completeness stability} is defined as the minimum eigenvalue of the scaled frame operator, and serves as a measure of how far the operator is from being invertible. Formally, we define
\begin{definition}
For any POVM $\mathcal A$, we let $s(\mathcal A)$ denote the minimum eigenvalue of $F_s(\mathcal A)$, and call it the \emph{completeness stability} of the POVM $\mathcal A$.
\end{definition}

We then have the following result:

\begin{proposition}
The completeness stability $\mathcal A\mapsto s(\mathcal A)$ is a resource monotone.
\end{proposition}
\begin{proof}
Let $\mathcal A, \mathcal B$ be POVMs such that $\mathcal A\succeq \mathcal B$. Let $C$ be any normalised eigenvector of $F_s(\mathcal A)$ corresponding to its minimum eigenvalue, so that $\langle C|F_s(\mathcal A)|C\rangle = s(\mathcal A)$. By Prop.~\ref{monotones},
\[
  s(\mathcal A)=\langle C|F_s(\mathcal A)|C\rangle \geq \langle C|F_s(\mathcal B)|C\rangle\geq s(\mathcal B),
\]
where the second inequality follows from~\eqref{scaled frame operator bound}.
\end{proof}

\subsection{Significance}

As discussed in Sec.~\ref{sec:frame-theory}, any quantum state \(\rho\) can be reconstructed from the set of probabilities \(\{p_{a}\}\) of a given POVM \(\mathcal{A} = \{A_a\}\) using any dual \(\{A'_a\}\) by \(\rho =\sum_{a}p_{a}A'_a\).
However, in practice only a finite number \(n\) of measurement results are available, and the probabilities are estimated from the frequencies of each outcome \(p_{a}\approx f_{a} = n_{a} /n\) which follow a multinomial distribution. Using the frequencies forms a linear unbiased estimator for state, \(\hat{\rho} = \sum_{a} f_{a} A'_a\). The estimation error can be characterized using the mean squared error and for a given dual frame is \cite{2011zhu327,2023inn328}:
\begin{equation}
  \mathcal{E}_{\rho} \equiv E[\lVert\hat{\rho}-\rho\rVert_{2}^{2}] = \frac{1}{n} \left(\sum_{a} p_a \mathrm{Tr}\left[ \ketbra{A'_a}{A'_a} \right] - \braket{\rho}{\rho}\right).
\end{equation}

To assess the quality of a measurement independently of the state, the estimation error is averaged over all unitarily equivalent states (which will have a fixed purity \(\mathcal{P} = \mathrm{Tr}[\rho^2]\)), resulting in~\cite{2011zhu327}:
\begin{equation}
  \overline{\mathcal{E}_{\mathcal{P}}} \equiv \frac{1}{n} \left( \sum_{a}  \frac{\mathrm{Tr}[A_a]}{d} \mathrm{Tr}\left[ \ketbra{A'_a}{A'_a} \right] - \mathcal{P} \right).
\end{equation}
The dual that minimizes this error is given as~\cite{2006sco507, 2011zhu327}:
\begin{equation}
\ket{A'_a}  = F_s(\mathcal{A})^{-1} \frac{\ket{A_a}}{\braket{A_a}{I}}.
\end{equation}
Consequently the average estimation error over unitarily equivalent states with a given purity becomes,
\begin{equation}
  \overline{\mathcal{E}_{\mathcal{P}}} = \frac{1}{n} \left(  \frac{1}{d} \mathrm{Tr}\left[ F_s(\mathcal{A})^{-1} \right] - \mathcal{P} \right).
\end{equation}

We can bound \(F_s(\mathcal{A})^{-1}\) by its largest eigenvalue, i.e. \(1/s(\mathcal{A})\) 
\begin{equation}
\mathrm{Tr}[F_s^{-1}(\mathcal{A})] \le \left( \frac{d^2 - 1}{s(\mathcal{A})} + 1 \right),
\end{equation}
leading to a bound on the average estimation error:
\begin{equation}
\frac{1}{n} \left( d^2 + d - 1 - \mathcal{P} \right)  \leq \overline{ \mathcal{E}_{\mathcal{P}}} \leq \frac{1}{n} \left( \frac{d^2 - 1}{d \, s(\mathcal{A})} + \frac{1}{d} - \mathcal{P} \right),
\end{equation}
where the lower bound was established by minimizing the trace of the inverse scaled frame operator~\cite{2006sco507}.
Thus, a higher completeness stability implies a lower estimation error. Note that as the completeness stability is increased both bounds approach each other, and for measurements that achieve maximal completeness stability (see next section) the bounds become equal.

Performing full quantum state tomography can be extremely resource-intensive as the dimension of the state get large, and in many cases only partial information about the state is desired. Another important application of completeness stability arises in the context of \emph{shadow tomography}, which aims to estimate the expectation values of observables to a desired accuracy without reconstructing the entire state~\cite{2020Hua057,2017aar053}.

The authors of~\cite{2023inn328} consider a similar optimisation over duals to that presented above to reduce the estimation variance of the expectation value of an observable \(\mathcal{O}\) on a quantum state \(\rho\). Using a fixed POVM \(\mathcal{A}\) and optimal associated dual frame, and again averaging over all unitarily equivalent quantum states, the average variance of estimation is bound by :
\begin{equation}
 \overline{\mathrm{Var}[\hat{o}|\mathcal{P}]} \leq \frac{V}{d} \left[ \frac{1}{s(\mathcal{A})} - \frac{d \, \mathcal{P} - 1}{d^2 - 1} \right],
\end{equation}
where \(V = \mathrm{Tr}[\mathcal{O}^2] - \mathrm{Tr}[\mathcal{O}]^2 /d\), and \(\mathcal{P}\) is the purity of the state.
Since all other parameters are measurement-independent, the estimation variance is primarily determined by the completeness stability again.

Finally, the completeness stability is also central to the numerical stability in the application of the scaled frame operator. This involves inverting the scaled frame operator, where numerical stability---that is, the sensitivity of the output to slight input perturbations---is quantified by the condition number of the matrix~\cite{1997TreNLA}. A lower condition number indicates greater numerical stability and ease of inversion. Since the scaled frame operator always has a maximum eigenvalue of \(1\), the condition number is the inverse of the completeness stability. A larger completeness stability implies a better numerical stability in inverting the frame operator.

\subsection{Extremes of the completeness stability}
\label{sec:extremes}

The completeness stability has a lower bound of zero, which corresponds to the case where the scaled frame operator becomes non-invertible. In this situation, the POVM is informationally incomplete and ceases to form a frame. To explore the other extreme, we seek to determine the \emph{maximum possible} value of the minimum eigenvalue \(\lambda_{\min}\) of the scaled frame operator.

To do so, we begin by analyzing the trace of the scaled frame operator. Using the inequality that the trace of a positive operator squared is less than or equal to the square of its trace, we have for each POVM element \(A_a \geq 0\):
\begin{align}
\mathrm{Tr}[F_s(\mathcal{A})] 
&= \sum_a \frac{\braket{A_a}{A_a}}{\braket{A_a}{I}} 
\leq \sum_a \frac{\braket{A_a}{I}^2}{\braket{A_a}{I}} \nonumber \\
&= \braket{\sum_a A_a}{I} = \mathrm{Tr}[I] = d,
\end{align}
where \(d\) is the dimension of the Hilbert space. Equality in this bound, i.e., \(\mathrm{Tr}[A_a^2] = \mathrm{Tr}[A_a]^2\), is achieved if and only if each \(A_a\) is a rank-one operator, \(A_a = w_{a}\ketbra{\phi_{a}}{\phi_{a}}\), with \(\sum_{a}w_{a}=d\). Note that the scaled frame operator will consequently be
\begin{equation}
  F_s(\mathcal{A}) = \sum_{a}w_{a}(\ketbra{\phi_{a}}{\phi_{a}})^{T}\otimes\ketbra{\phi_{a}}{\phi_{a}},
\end{equation}
where \(T\) denotes a transpose.

Since the largest eigenvalue is always 1 (corresponding to the normalised eigenvector \(\ket{I} /\sqrt{d}\)), and the trace is fixed to be \(d\), the maximum possible value of \(\lambda_{\min}\) is achieved when the remaining \(d^2 - 1\) eigenvalues are all equal. In this extremal case, those eigenvalues must each be equal to \(1/(d+1)\), leading to the scaled frame operator having the structure:
\begin{equation}\label{scaled frame for maximal condition}  
F_s(\mathcal{A}) = \frac{\ketbra{I}{I}}{d} + \sum_{i=1}^{d^2 - 1} \frac{\ketbra{\mathcal{X}_i}{\mathcal{X}_i}}{d+1},
\end{equation}
where \(\{ \ket{\mathcal{X}_i} \}\) form an orthonormal basis for the traceless subspace of \(\mathcal{L}(\mathcal{H})\).
We can rewrite \(F_s(\mathcal{A})\) in terms of projectors into each eigenspace as
\begin{align}\label{eq:tight-rank1-IC}
  F_s(\mathcal{A}) &= \frac{1}{d} \ketbra{I}{I} + \frac{1}{d+1} \left(I_{d^{2}}-\frac{1}{d}\ketbra{I}{I}\right) \nonumber\\
  &= \frac{1}{d+1}\left(I_{d^{2}}+\ketbra{I}{I}\right).
\end{align}
Where $I_{d^{2}}$ is the $d^{2}$-dimensional identity matrix. POVMs whose scaled frame operator matches this form are known as \emph{tight rank-1 informationally complete POVMs} \cite{2006sco507}. For these POVMs, the set \(\{w_{a}, \ket{\phi_{a}} \}\) forms a 
weighted complex projective 2-design. 
The POVMs on this boundary exhibit not only extremal resource-theoretic properties but also optimal performance in practical applications, including linear quantum state tomography and quantum cloning~\cite{2006sco507}.

Two well-known examples of such boundary case POVMs include the \emph{Symmetric Informationally Complete POVMs} (SIC-POVMs) and POVMs constructed from sets of \emph{Mutually Unbiased Bases} (MUBs). We will explicitly compute the completeness stability for these examples below as the proofs are simple and instructive.

A Symmetric Informationally Complete POVM (SIC-POVM) \(\mathcal{E}=\{ E_i \}\), is a highly symmetric set of rank-1 effects that is minimally informationally complete, and satisfies the following two properties:  
\begin{equation}
\text{Tr}[E_i] = \frac{1}{d}, \quad \quad \text{Tr}[E_i E_j] = \frac{d \delta_{ij} + 1}{d^2 (d+1)}.
\end{equation}
They are typically optimal in various tasks involving measurement \cite{2014app484, 2008Dur338, 2006sco507, 2006kim024, 2005ren81, 2004reh321} and are of interest in their own right \cite{2017Fuc021}.

\begin{proposition}
  SIC-POVMs have maximal completeness stability.
\end{proposition}
\begin{proof}
For a SIC-POVM the scaled Gramian matrix will be \( d^2\times d^2 \) and give by
\begin{equation}
  \label{eq:Gsexplicit}
G_{s}(\mathcal{E}) =\sum_{j,k} \frac{\text{Tr}[E_j E_k]}{\sqrt{\text{Tr}[E_j] \text{Tr}[E_k]}}\ketbra{j}{k},
\end{equation}
it will have \(1/d\) on the diagonal and \(1/(d(d+1))\) elsewhere. \(G_{s}(\mathcal{E})\) will have the same set of non-zero eigenvalues as the scaled frame operator as shown in section~\ref{sec:frame-theory}.
We can decompose $G_{s}(\mathcal{E})$ as:
\begin{equation}\label{eq:Gsic}
    G_{s}(\mathcal{E}) = \frac{1}{d+1} I_{d^{2}} + \frac{1}{d(d+1)} J_{d^{2}},
\end{equation}
where $I_{d^{2}}$ is the \(d^{2}\)-dimensional identity matrix and
$J_{d^{2}}$ is the \(d^{2}\)-dimensional matrix with all entries equal to 1.
The matrix $J_{d^{2}}$ is a rank-1 matrix with known eigenvalues of  \(d^2\) and 0 with degeneracies of 1 and $d^2 - 1$ respectively.
Thus, the eigenvalues of $J_{d^{2}}/({d(d+1)})$ will be \( {d}/({d+1}) \) and $0$.
Since addition of the \(I_{d^{2}}\) term in \eqref{eq:Gsic} just shifts the eigenvalues by ${1}/({d+1})$,
the eigenvalues of \(G_{s}(\mathcal{E})\) are 1 (non-degenerate) and
\({1}/{(d+1)}\) (with a degeneracy of $d^2 - 1$). The frame operator
\(F_{s}(\mathcal{E})\) will have the same eigenvalues. Therefore SIC-POVMs have maximum completeness stability.
\end{proof}

Mutually unbiased bases (MUB) are sets of vectors that also have interesting symmetry properties and many applications, see for example \cite{2010dur535}. A MUB is a set of orthonormal bases such that for any pair of bases in the set, $\{ \ket{m_{1j}} \}$ and $\{ \ket{m_{2k}} \}$ say, then we have $ \braket{m_{1i}}{m_{1j}}=\braket{m_{2i}}{m_{2j}}= \delta_{i,j}$, and  $|\braket{m_{1i}}{m_{2j}}|^{2}={1}/{d}$. We can have up to \(d+1\) such bases in a set, and a set that has \(d+1\) bases is called \emph{maximal}. Maximal MUB are only known when the dimension is a power of a prime (i.e. \(d=p^{n}\) where \(p\) is prime).

We can construct a POVM out of a MUB, which we call a MUB-POVM, by
taking the projectors for all the vectors and scaling them
appropriately. A \emph{maximal} MUB-POVM will have the $(d+1)$ effects: 
\begin{equation}
  \mathcal{M}=\{\Pi_{11}, \dots, \Pi_{1d}, \dots, \Pi_{(d+1)1}, \dots ,\Pi_{(d+1)d}\}, 
\end{equation}
where \(\Pi_{jk}=\ketbra{m_{jk}}{m_{jk}}/(d+1)\).

\begin{proposition}
  Maximal MUB-POVMs have maximal completeness stability. Furthermore, a non-maximal
  MUB-POVM is informationally incomplete.
\end{proposition}
\begin{proof}
 For a maximal MUB-POVM, the Gramian matrix \eqref{eq:Gsexplicit} takes the form of a \(d(d+1)
 \times d(d+1)\) matrix 
\begin{align*}
   G_{s}(\mathcal{M}) &= \frac{1}{d+1}I_{d(d+1)} +\frac{1}{d(d+1)} (J_{d+1} - I_{d+1}) \otimes J_{d}.
\end{align*}

To determine the eigenvalues of \(  G_{s}(\mathcal{M}) \), we first analyze the spectral
properties of its components. The matrix \( J_{d+1} \) has eigenvalues
\( d+1 \) with multiplicity 1, and 0 with multiplicity \( d \). Consequently \( J_{d+1} - I_{d+1} \) has eigenvalues \( d \) with multiplicity 1 and \( -1 \) with multiplicity \( d \).  

The eigenvalues of a tensor product of matrices are given by the
products of their individual eigenvalues so \((J_{d+1} - I_{d+1}) \otimes J_{d} \)  
has eigenvalues given by \( d^2 \) with multiplicity 1, \( -d \) with
multiplicity \( d \), and 0 with multiplicity \( d^2 - 1 \).
Multiplying by the factor \( {1}/{(d(d+1))} \), and including the
shift from the term \( \frac{1}{d+1}I_{d (d+1)}\) results in
eigenvalues of the Gramian matrix \( G_{s}(\mathcal{M}) \) of \( 1 \)
with a multiplicity of 1, \( {1}/{(d+1)} \) with a multiplicity of
\( d^2 - 1 \), and eigenvalue \( 0 \) with a multiplicity of \( d \).
Since the corresponding frame operator \(F_{s}(\mathcal{M})\) has
\(d^{2}\) eigenvalues they are all non-zero and the minimum is
\({1}/{(d+1)}\).

Now consider a non-maximal MUB with say \(d\) sets of bases. Repeating
the calculation above we find there are only \(d^{2}-d+1\) non-zero
eigenvalues, consequently \(F_{s}(\mathcal{M})\) has \(d-1\) zero
eigenvalues and is not invertible --- i.e. the POVM is informationaly incomplete.
\end{proof}

\subsection{Product measurements}

Consider a bipartite system where one party implements the POVM \(\mathcal{A}\) and the other party implements POVM \(\mathcal{B}\). The scaled frame operator for the joint product-measurement becomes
\begin{align}
  F_{s}\left(\mathcal{A}\otimes \mathcal{B}\right) &= \sum_{j,k} \frac{\ketbra{A_{j}\otimes B_{k}}{A_{j}\otimes B_{k}}}{\braket{A_{j}\otimes B_{k}}{I}}\\
                                                   &= S_{2,3} \sum_{j,k} \frac{\ketbra{A_{j}}{A_{j}}\otimes\ketbra{B_{k}}{B_{k}}}{\braket{A_{j}}{I}\braket{B_{k}}{I}} S_{2,3}\\
  &= S_{2,3} (F_{s}(\mathcal{A})\otimes F_{s}(\mathcal{B})) S_{2,3},
\end{align}
where \(S_{2,3}\) is the swap operator between spaces 2 and 3. Since the swap does not change the eigenvalues,
\begin{equation}
  s(\mathcal{A}\otimes \mathcal{B}) = s(\mathcal{A})s(\mathcal{B}).
\end{equation}

The result generalises to \(n\)-partite product measurements where the completeness stability becomes the product of the individual completeness stabilities. This indicates that the best product measurements, in the sense of maximum completeness stability, will become exponentially worse with the number of parties \(n\) compared to the best joint measurement.

\section{Examples}
\label{sec:examples}

Consider a process that degrades a POVM \(\mathcal{A}=\left\{A_{1},
  \dots,A_{a}\right\}\) by adding a portion of the identity operator
to each effect to produce POVM \(\mathcal{B}(p)\) where
\begin{equation}
  B_{j}(p) = (1-p)A_{j}+\frac{p}{a}I.
\end{equation}

The completeness stability of \(\mathcal{B}(p)\) will be a monotonically
non-increasing function of \(p\). To see this, write the identity as
\(I=\sum_{j}A_{j}\) and now \(\vec{B}(p)=S(p)\vec{A}\) where \(S(p)\)
is a column stochastic matrix with \((1-p)+p/a\) on the diagonal and
\(p/a\) elsewhere. Furthermore this degradation composes as
\(S(p_{1})S(p_{2})=S(p_{3})\) with \(p_{2}=(p_{3}-p_{1})/(1-p_{1})\).
Since the completeness stability is non-increasing under stochastic
post-processing, the completeness stability of \(\mathcal{B}(p)\) will be a monotonically
non-increasing function of \(p\). The robustness of the measurement is
linearly decreasing for this degradation.
\begin{align}
  R(\mathcal{B}(p)) &= \sum_{j}\Lambda_{\max}(B_{j})-1 \nonumber\\
                    &= \sum_{j}\left[(1-p)\Lambda_{\max}(A_{j})+\frac{p}{a}\right]-1 \nonumber\\
                    &=(1-p) R(\mathcal{A})
\end{align}

For instance, in \( d=2 \), consider the following parameterized POVM
with six effects:  
\begin{align}\label{eq:povmEx1}
\left\{ \frac{1-p}{3}\ketbra{\pm x}{\pm x} + \frac{p}{6}I, \frac{1-p}{3}\ketbra{\pm y}{\pm y} + \frac{p}{6}I \right.\nonumber\\
\left.\frac{1-p}{3}\ketbra{\pm z}{\pm z} + \frac{p}{6}I  \right\},
\end{align}
where \(  \ket{\pm x} \), \(\ket{ \pm y} \), \(  \ket{\pm z} \) are the eigenstates of \( \sigma_x \), \( \sigma_y \), and \( \sigma_z \), respectively.

This set satisfies the POVM conditions, as the effects sum to the identity and remain positive semi-definite for \( 0 \leq p \leq 1 \). At \( p = 0 \) each effect is a rank-1 projector, whereas at \( p = 1 \) the POVM is completely uninformative. As expected, both the robustness and the completeness stability become worse as \( p \) increases, see Fig.~\ref{fig:ex1}. 

\begin{figure}[h]
  \centering
  \includegraphics[width=0.45\textwidth]{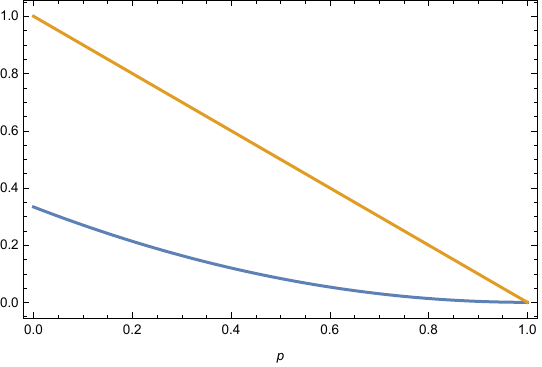}
  \caption{Example 1 with POVM \eqref{eq:povmEx1} showing the
    completeness stability (blue) and the robustness (orange).}
  \label{fig:ex1}
\end{figure}

A second example highlights a scenario where the completeness stability provides additional insights beyond robustness. Consider another parameterized POVM in \( d=2 \), given by  
\begin{equation}\label{eq:povmEx2}
\left\{\frac{1-a}{2}\ketbra{\pm x}{\pm x},\frac{1-a}{2}\ketbra{\pm y}{\pm y},a\ketbra{\pm z}{\pm z}\right\},
\end{equation}
where \( 0\leq a \leq 1 \), and \( \ket{\pm x} \), \( \ket{\pm y} \), \( \ket{\pm z} \) are the eigenstates of \( \sigma_x \), \( \sigma_y \), and \( \sigma_z \), respectively.  

Intuitively, at \( a = 0 \), measurements are only performed along the \( x \)- and \( y \)-directions in the Bloch sphere. As \( a \) increases, the weights of the \( x \)- and \( y \)-effects decrease, while the weight in the \( z \)-direction increases. At \( a = 1 \), the measurement consists solely of effects in the \( z \)-direction. Thus, the most balanced or symmetric measurement is expected to be optimal for reconstructing an arbitrary state.  

For this POVM, the robustness remains constant at 1 for all values of \( a \), making it an ineffective measure of quality.  However, the completeness stability is given by $\min(a,{(1-a)}/{2})$, which better captures the measurement’s ability to reconstruct a state.
\begin{figure}[h]\label{fig:ex2}
  \centering
  \includegraphics[width=0.45\textwidth]{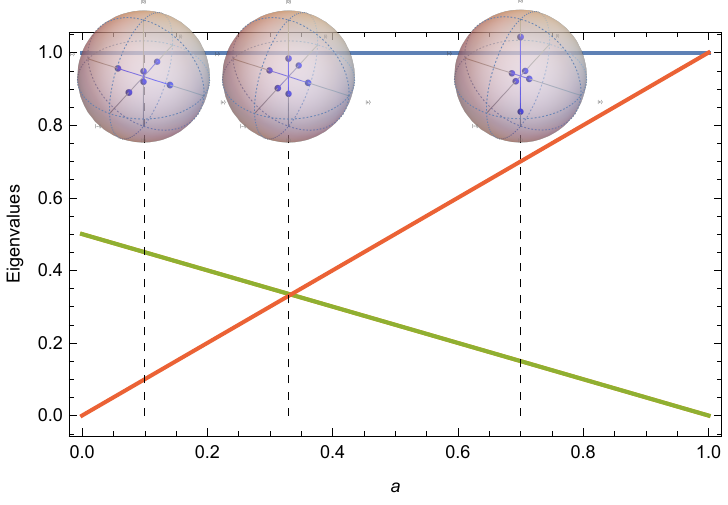}
  \caption{Example 2 with POVM \eqref{eq:povmEx2}, showing the
    eigenvalues of \(F_{s}\). Note that the maximum completeness stability occurs
    at \(a={1}/{3}\). Insets show the POVM effects plotted on the
    Bloch ball for \(a=\{0.1,0.333,0.7\}\). At \(a={1}/{3}\) the
    effects are most evenly distributed, and also achieving maximum completeness stability, and in the limits of \(a=0\)
    and \(a=1\) the POVM becomes informationally incomplete signalled
    by the presence of zero eigenvalues.}
\end{figure}

\section{Discussion and Conclusion}

In this work, we introduced an informative and operationally meaningful resource monotone for assessing the quality of a quantum measurement. Within this resource framework, quantum measurements are partially ordered under classical post-processing which is quantified by resource monotones. Ideally, one would like a \emph{complete} set of monotones \(\{\mu_j\}\) such that \(\mu_j(\mathcal{A}) \ge \mu_j(\mathcal{B})\) for all \(j\) implies \(\mathcal{A} \succeq \mathcal{B}\). While the set of all quantum state discrimination tasks forms such a complete family of monotones, it is vast and often impractical to work with. This highlights the need for compact, physically motivated monotones.

We focused on the regime of informationally complete POVMs, which are central to quantum information theory. We employed frame theory and introduced a new class of monotones derived from the \emph{scaled frame operator}. Among these, the minimum eigenvalue, which we referred to as the \emph{completeness stability}, emerged as a particularly significant quantity. It characterizes how close a POVM is to being informationally incomplete and also serves as an indicator of numerical stability. 

Beyond its role as a resource monotone, completeness stability also has direct operational significance in quantum information tasks. In particular, the state estimation error and variance of observable estimation in shadow tomography, averaged over all states of fixed purity, is tightly upper bounded by the inverse of completeness stability. Thus, POVMs with higher completeness stability yield more accurate and stable estimation procedures.

Interestingly, the highest possible value of this monotone is a signature of well-studied classes of measurements, corresponding to 2-designs in complex projective space. These are optimal for linear state tomography and also possess deep structural and symmetry properties. The completeness stability then offers potential avenues for further exploration such as approximately optimal POVMs, that have nearly maximal completeness stability.

Overall, our work highlights completeness stability as a meaningful and versatile monotone that connects resource-theoretic considerations with concrete operational advantages. It offers a principled way to compare and classify POVMs, and is a tool to improve the quality of a POVM.

A promising direction for future work is to identify new resource monotones that are both operationally meaningful and tailored to other specific quantum information tasks. Another compelling avenue is to extend the current framework beyond POVMs to encompass sets of quantum measurements, and to develop corresponding resource theories capable of capturing a wider range of operational scenarios.

\begin{acknowledgments}
  This research was funded in part by the Australian Research Council Centre of Excellence for Engineered Quantum Systems (Project number CE170100009). R.S. acknowledges funding from Sydney Quantum Academy.
\end{acknowledgments}

\bibliography{references}

\appendix

\section{Monotones based directly on frame operators}
\label{frame monotone}

Interestingly it seems much harder to find monotones based directly on
the frame operator itself~\eqref{eq:frame-operator}. We find that the
determinant of the frame operator is a monotone if the frame is
\emph{exact}, that is the POVM has \(d^{2}\) linearly independent effects.
\begin{proposition}
If the frame formed by the POVM effects of \(\mathcal{A}\) is exact, the determinant of the frame operator, $\det[F(\mathcal{A})]$, is a resource monotone.
\end{proposition}

\begin{proof}
  Consider a POVM \(\mathcal{A}\) with \(d^{2}\) linearly independent effects, in which case they form an exact frame. We can construct the rank-1 parent matrix \(M(\mathcal{A})\) \eqref{eq:Mv} from which we can determine both the frame operator \(F(\mathcal{A})\) and the Gramian \(G_{s}(\mathcal{A})\). Since the frame is exact, both of these matrices will have the same dimensions and eigenvalues.
      
Now for another POVM $\mathcal{B}$ with $\mathcal{A} \succeq  \mathcal{B}$, there exists a column-stochastic matrix $S$ with $\vec{B}=S\vec{A}$. If we expand out the symbolic representation in the vectors such that the first $d^{2}$ elements are the elements of $\ket{A_{1}}$, the second $d^{2}$ elements are from $\ket{A_{2}}$ and so on, we want the stochastic matrix to act in block-form, i.e. $S\otimes I$. That is,
\begin{equation}
    \sum_{j}\ket{j}\ket{B_{j}} = (S\otimes I) \sum_{k}\ket{k}\ket{A_{k}}.
\end{equation}
We now have
\begin{align}
    M(\mathcal{B})&= \frac{1}{N_{B}}\sum_{a,d}(\ketbra{a}{d}\otimes\ketbra{B_{a}}{B_{d}})\nonumber\\
                      &= \frac{1}{N_{B}}(S\otimes I)\sum_{a,d}(\ketbra{a}{d}\otimes\ketbra{A_{a}}{A_{d}})(S^{\dagger}\otimes I)\nonumber\\
    &= \frac{N_{A}}{N_{B}}(S\otimes I)M(\mathcal{A})(S^{\dagger}\otimes I)\label{eq:mvAB}
\end{align}

Since $M$ is a bi-partite pure state, the eigenvalues of the reduced density matrices will be equal so that
\begin{equation}
    \det[\tr_{1}M(\mathcal{B})]=\frac{1}{N_{B}^{d^{2}}}\det[F(\mathcal{B})]=\det[\tr_{2}M_{v}(\mathcal{B})].
\end{equation}
and using~\eqref{eq:mvAB}
\begin{align}
    \det[F(\mathcal{B})]&=N_{B}^{d^{2}}\det\left[\tr_{2}\left(\frac{N_{A}}{N_{B}}(S\otimes I)M(\mathcal{A})(S^{\dagger}\otimes I)\right)\right]\nonumber\\
                        &=N_{A}^{d^{2}}\det[{S}]\det[\tr_{2}M(\mathcal{A})]\det[S^{\dagger}]\nonumber\\
    &=N_{A}^{d^{4}}\det[S^{\dagger}S]\det[\tr_{1}M(\mathcal{A})]\nonumber\\
    &=\det[S^{\dagger}S]\det[F(\mathcal{A})]
\end{align}

Since \( S \) is a column-stochastic matrix, all of its eigenvalues lie within the unit circle in the complex plane; that is, \( |\lambda_s| \le 1 \) for all eigenvalues \( \lambda_s \) of \( S \). Consequently, \( \det[S^\dagger S] = |\det[S]|^2 \le 1 \). Hence, \( \det[F(\mathcal{B})] \le  \det[F(\mathcal{A})]\)

\end{proof}

\section{Counter example for complete set of monotones}
\label{sec:counter}

In proposition~\ref{monotones} we introduced an infinite set of resource monotones \(\mu_{X}(\mathcal{A})=\bra{X}F_{s}(\mathcal{A})\ket{X}\). It might be hoped that this set is a \emph{complete} set, but unfortunately this is not the case, and we show how to construct counter examples in this section.

A complete monotone set has the property
\begin{equation}
  \mu_{X}(\mathcal{A})\geq \mu_{X}(\mathcal{B})\;\;\forall X \implies \mathcal{A}\succeq \mathcal{B}.
\end{equation}
We have already established that
\begin{equation}
  \mathcal{A}\succeq \mathcal{B} \implies \bra{X}F_{s}(\mathcal{A})-F_{s}(\mathcal{B})\ket{X}\geq 0 \;\;\forall X,
\end{equation}
consequently \(\Delta F_{s}=F_{s}(\mathcal{A})-F_{s}(\mathcal{B})\) will be a positive semi-definite operator. N.B. it will always have a 0 eigenvalue since \(\ket{I}\) is always an eigenvector. If we find an example where \(\mathcal{A}\nsucceq \mathcal{B}\) but \(\Delta F_{s}\) is also positive semi-definite, then the set \(\mu_{X}(\mathcal{A})\) cannot be a complete set of monotones.

We start with an exact informationally complete POVM \(\mathcal{A}\) (one that has \(d^{2}\) elements which span the operator space). Since the operator frame is exact, it will form a basis. Now construct another POVM \(\mathcal{B}\):
\begin{equation}
  \vec{B}=R \vec{A}
\end{equation}
where \(R\) is a real matrix whose columns sum to 1 but include negative entries. Note that \(\vec{B}\) will not automatically form a POVM, that will depend on the values of \(\vec{A}\) and \(R\). It is always possible to permute the elements of a POVM but combining \(R\) with any permutation would still leave a matrix with negative entries. Now we can numerically search for \(R\) matrices that yield a valid POVM \(\mathcal{B}\) with distinct effects, and that has \(\Delta F_{s}\geq 0\). The initial POVM \(\mathcal{A}\) can also be generated randomly in the search for counter examples, following for example the procedures in~\cite{2020hei202}. 

Counter examples can be readily found, below we give an explicit example with exact values. Consider the following states for a qubit:
\begin{align}
  \ket{\psi_{1}} &= \ket{0}\\
  \ket{\psi_{2}} &= (\ket{0}+\sqrt{2}\ket{1})/\sqrt{3}\\
  \ket{\psi_{3}} &= (\ket{0}+\sqrt{2}e^{2\pi i /3}\ket{1})/\sqrt{3}\\
  \ket{\psi_{4}} &= (\ket{0}+\sqrt{2}e^{4\pi i /3}\ket{1})/\sqrt{3}.
\end{align}
Construct the effects of a SIC-POVM \(\mathcal{A}\) as \(A_{j} =\ketbra{\psi_{j} }{\psi_{j} }/2\). Now set,
\begin{equation}
  R = \frac{1}{9}
  \begin{pmatrix}
    10 & 3 & 0 & 0 \\
    -1 & 3 & 6 & 0 \\
    0 & 3 & 3 & 0 \\
    0 & 0 & 0 & 9
  \end{pmatrix}.
\end{equation}
This generates a valid informationally complete POVM \(\mathcal{B}\) and the eigenvalues of \(\Delta F_{s}\) are \(\{ (57+\sqrt{2053})/312 ,(57-\sqrt{2053})/312, 0, 0\}\) all of which are non-negative. 

\end{document}